\documentclass{amsart}
\usepackage{amsmath}%
\usepackage{amsfonts}%
\usepackage{amssymb}%
\usepackage{graphicx}
\newtheorem{theorem}{Theorem}
\theoremstyle{plain}

\numberwithin{equation}{section}
\begin{document}
\title[On the solution of Four Color problem]{Special partial graphs}
\author{Natalia Malinina}
\address[Moscow Aviation Institute (National Research University), 4, Volocolamskiy Shosse, Moscow, GCP-4, Russian Federation]
\newline%
\email[Natalia Malinina]{malinina806@gmail.com}%


\thanks{This paper is in preliminary form and it will be submitted for publication in some journal.}
\date{December, 2012}
\dedicatory{Dedicated to my father Leonid Malinin}

\begin{abstract}
The attempts to prove the Four Color Problem last for years. A little hope arises that the properties of the minimal partial triangulations will be very useful for the solution of the Four Color Problem. That is why the material of this paper is devoted to the examination of the specific partial graphs and their properties. Such graphs will have all the elements of the planar conjugated triangulation, but will have the minimal size. And it will be quite interesting to find out their properties in order to search in the sequel for the possibility to prove the Four Color Problem on the base of their characteristics.  
\end{abstract}
\maketitle

\section{Introduction}

Right along it was used to examine the arbitrary graph's properties basing on the properties of some partial graphs, which have similar characteristics as the graphs in question. So the properties of the planar triangulations will be better to examine on the base of the partial graphs of special type. And such graphs will be introduced and examined here.

\section{The minimal closed graph and its properties} 

Let us introduce graph $H_{min}$. Such graph has all the elements of the planar conjugated triangulation $H(V,Q)$, but has the minimal possible size (fig. \ref {Fig:1}).

\begin{figure}[htb]
		\includegraphics[width=0.3\textwidth]{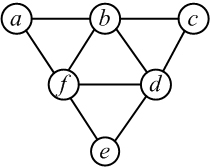}
	\caption{The closed minimal graph}
	\label{Fig:1}
\end{figure}

Graph $H_{min}$ has three faces of the first subset $\{f_{hi}^{(1)}\}$, one face of the second subset $\{f_{hi}^{(2)}\}$, $9$ edges and $6$ vertexes. It is evident that such a graph is unique. Let us prove its properties.

\begin{theorem}
Graph $H_{min}$ has five variants of the edges' orientation along the Euler circuit at the given direction of the pass through one external vertex.
\end{theorem}

\begin {proof}
Let us examine all the possible variants of graph's $H_{min}$ edges' orientation and choose those of them, which correspond to Euler circuits. Let us specify two arbitrary taken ordered pairs of the vertexes, for instance, $fa$ and $ab$. Thus, one exit from vertex $f$ is already given, and one entry into vertex $b$ is also given. As far as every vertex must have the equal numbers of both the entries and the exits, vertex $f$ can have one more exit out of the remaining three ones: $fb$, $fd$ or $fe$, and vertex $b$ --- one entry out of the remaining three ones: $fb$, $db$ or $cb$. 

Thus, we have $9$ assumed or supposed variants of all the combinations of the exits from vertex $f$ and the entries into vertex $b$. It is evident that all the possible variants of the edges' orientation in graph $H_{min}$, which may correspond to Euler circuits are exhausted. But some variants do not permit Euler circuits.

Let us make a table, which contains all nine variants. Four of nine variants (fig. \ref{Fig:2}) do not permit the existence of Euler circuit: $(12)$, $(13)$, $(21)$ and $(31)$. Another five variants permit the existence of Euler circuit: $(11)$, $(22)$, $(23)$, $(32)$ and $(33)$. Other variants at the given direction through the vertex $a$ are principally impossible. As far as only two passes' directions are possible through the external vertex, then the total amount of the variant's groups will be $10$ (taking into account $5$ variants of the edges' direction). 

\begin{figure}[htb]
		\includegraphics[width=1\textwidth]{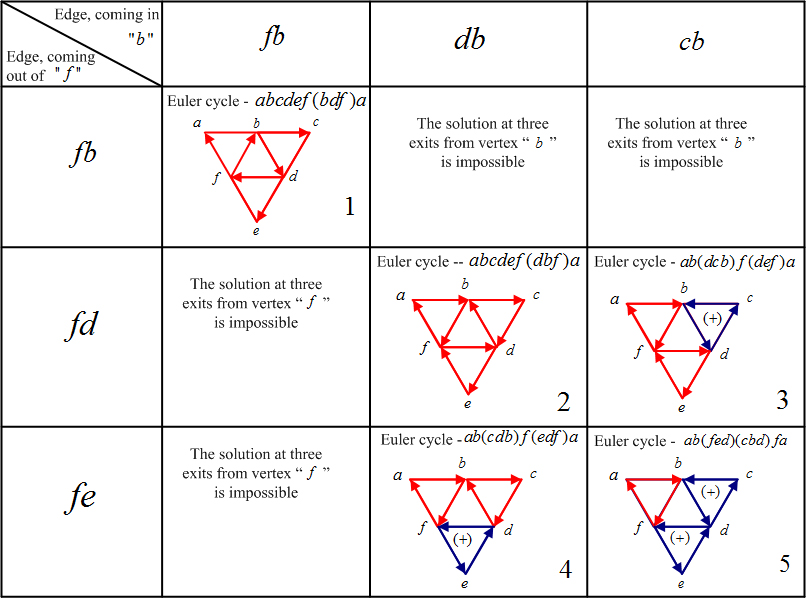}
	\caption{}
	\label{Fig:2}
\end{figure}

\end {proof}

Let us differentiate the variants of Euler circuits with the signs of the faces' bypasses $\{h_j \}_1$ as $(+)$ --- counterclockwise (blue) and as $(-)$ --- clockwise (red).

\begin {enumerate}
\item
External circuit --- (-); inner circuit --- (-);
\item
All the circuits --- (-);
\item
Two circuits ---  (-); one circuit --- (+);
\item
Two circuits --- (-); one circuit --- (+);
\item
Two circuits --- (+); one circuit --- (-);
\end {enumerate}

The foresaid is equal for the closed graph $H_{min}$. What we shall see if graph $H_{min}$ is the open-ended one? 

Let us also examine such opportunity.

\section {The minimal opened graph and its properties}

For the examining the properties of the opened graph let us introduce the minimal connected opened graph $\widetilde{H}_{min}$, which has three faces of the first set $\{f_{hj}^{(1)}\}$, no faces of the second set $\{f_{hj}^{(2)}\}$, $9$ edges and $7$ vertexes. Two vertexes are the cut points. This graph is presented in fig. \ref{Fig:4}.

\begin{figure}[htb]
		\includegraphics[width=0.4\textwidth]{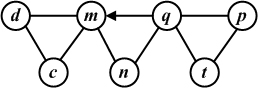}
	\caption{}
	\label{Fig:4}
\end{figure}

The connected graph $\widetilde{H}_{min}$ corresponds to some minimal part of graph $H$, which is cut out of graph $H$ and which is open-ended. Graph $\widetilde{H}_{min}$ has three faces. One face, let us denote it as the middle face, contains two vertexes, which are the cut points. Two other faces may be denoted as the last faces. Let us denote the edge, which is connecting the cut points in graph $\widetilde{H}_{min}$ as the main edge. The other $6$ edges will be denoted as the attached edges. As far as all graphs $\widetilde{H}_{min}$ are equal concerning their configuration, then for the examining of their properties it will be sufficient to examine only one out of the whole amount of them.

\begin {theorem}
The minimal graph $\widetilde{H}_{min}$ of the planar conjugated triangulation contains $9$ variants of the groups of Euler circuits, which pass through this subgraph.
\end {theorem}

\textit{Remark}: Under examining the possible variants of Euler circuits we will be interested only in those parts of the passes, which pass through the edges, which are incident to the cut points. 

\begin {proof}
Subgraph $\widetilde{H}_{min}$, which consists out of three faces $\{f_{hj}^{(1)}\}$, has two cut points (fig. \ref{Fig:4} --- vertexes $m$ and $q$). Let us agree to define the variants of Euler circuits with the help of setting the orientation only to the edges, which are incident to the cut points. Beforehand we will settle the orientation only by the edge, connecting the vertexes, which are the cut points.

Let us examine all the possible variants of the edges' orientation, taking as an example graph $\widetilde{H}_{min}$. We will examine all the variants of Euler circuits, which correspond to the direction of the edge connecting the cut points. All of them are presented in table (fig. \ref{Fig:5}). 

\begin{figure}[htb]
	\centering
		\includegraphics[width=1\textwidth]{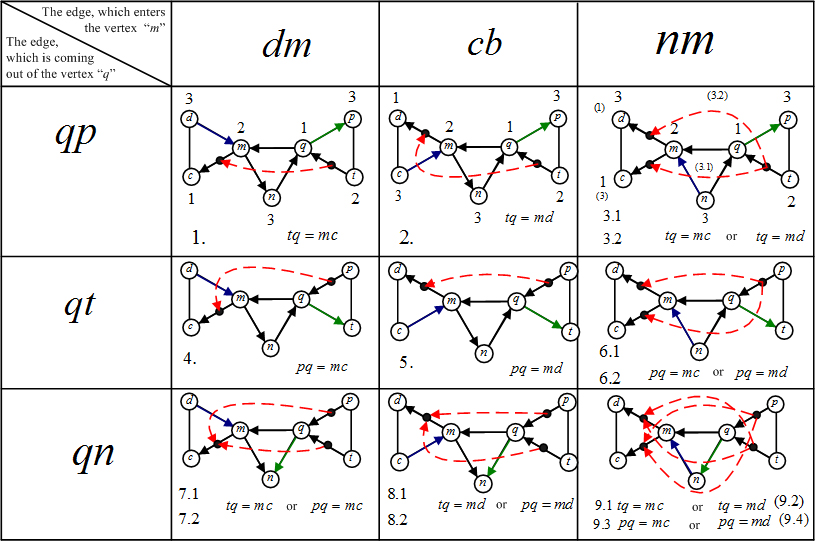}
	\caption{}
	\label{Fig:5}
\end{figure}

Let the edge $qm$ be given. So, it is given one exit from vertex $q$ and one entry into vertex $m$. In addition there are possible three variants of exits from vertex $q$ and three variants of the entries into vertex $m$. Thus, we have 9 expected variants of Euler circuits in graph $\widetilde{H}_{min}$. There none other variants of Euler circuits, which pass through subgraph $\widetilde{H}_{min}$.
\end {proof}

To make the visual interpretation more comfortable, the edges, entering the vertex $m$ are colored blue (the beginning of the cycle), and the edges, coming out of the vertex $q$ are colored green (the end of the cycle).

\section {The graph, which is conjugated to the minimal opened graph}

Let us also construct graph $G$, which will be the adjacency graph of the edge graph $H$. Let us examine subgraph $\widetilde{G}_{min}$ (fig. \ref{Fig:6}) of graph $G$. For more visualization subgraph $\widetilde{G}_{min}$ (on the left) is presented with much thinner lines. On the right it is presented with much fatter lines.

\begin{figure}[htb]
		\includegraphics[width=0.8\textwidth]{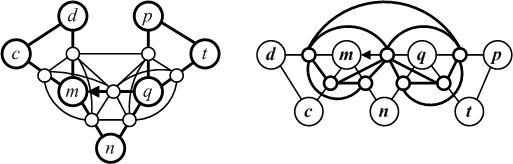}
	\caption{The graph, which is conjugated to $\widetilde{H}_{min}$}
	\label{Fig:6}
\end{figure}

Let subgraph $\widetilde{G}_{min}$ be an adjacency graph for only those $7$ edges of graph $\widetilde{H}_{min}$, which are incident to the cut points. All these edges are oriented along this or that Euler circuit. So each of them may be examined as an ordered pair of vertexes of graph $\widetilde{G}_{min}$. According to this every vertex of graph $\widetilde{G}_{min}$ may be indicated as an ordered pair of the vertexes of graph $\widetilde{H}_{min}$. Similarly to graph $\widetilde{H}_{min}$ we will denote one vertex of graph $\widetilde{G}_{min}$ as the main vertex, and the other vertexes --- as the attached vertexes.

Graph $\widetilde{G}_{min}$, which is conjugated to the minimal graph $\widetilde{H}_{min}$ has the following elements according to the construction:

\begin {itemize}
\item 
$3$ vertexes, appropriated to the edges of the middle face of the graph $\widetilde{H}_{min}$ (one vertex has a degree equal to $6$, two others --- $4$);
\item 
$4$ vertexes, appropriated to the edges of the last faces of graph $\widetilde{H}_{min}$ (all of them have the degree equal to $3$);
\item 
$3$ edges, connecting $3$ vertexes, appropriated to the edges of the middle face of graph $\widetilde{H}_{min}$;
\item 
$2$ edges, connecting $2$ vertexes, appropriate to the edges of last face of graph $\widetilde{H}_{min}$;
\item 
$8$ edges, each connecting two vertexes, appropriated one --- to the last face, and the other --- to the middle face of graph $\widetilde{H}_{min}$.
\end{itemize}

Altogether graph $\widetilde{G}_{min}$ has $7$ vertexes and $13$ edges. As graph $\tilde{H}_{min}$ always has the same configuration, graph $\widetilde{G}_{min}$ also always has the same configuration. Let us accept an arbitrary notation of graph's $\widetilde{H}_{min}$ vertexes, and denote the vertexes of graph $\widetilde{G}_{min}$ as the ordered pairs of graph's $\widetilde{H}_{min}$ vertexes in compliance with the variant of the pass along Euler circuit in graph $\widetilde{H}_{min}$.

Let us compose the adjacent matrix of graph's $\widetilde{G}_{min}$ vertexes. The rows will be arranged in the order of the clockwise bypass of graph $\widetilde{G}_{min}$, beginning from the vertex, which is corresponding to the edge of graph $\widetilde{H}_{min}$ connecting the cut points of graph $\widetilde{H}_{min}$. The matrix will always have the same appearance as it is presented in fig. \ref{Fig:7}, because graph $\widetilde{G}_{min}$ will always have 13 edges, which connect quite definite vertexes of the graph.

\begin{figure}[htb]
		\includegraphics[width=0.4\textwidth]{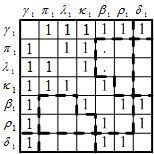}
	\caption{}
	\label{Fig:7}
\end{figure}

\begin{theorem}
Each of graph's $\widetilde{G}_{min}$ vertexes always can be signed with the help of the pair of three numbers ($1$, $2$ and $3$) thereby none of the adjacent ones will be signed equally.
\end {theorem}

\begin {proof}
Six pairs can be composed out of three numbers ($1$, $2$ and $3$). Each pair can be presented with some color. The maximal degree of graph's $\widetilde{G}_{min}$ vertex is not more than $6$. So the chromatic number of graph $\widetilde{G}_{min}$ is not more than $6$ (it follows from Brook's theorem).
\end {proof}

\begin {theorem}
Let it be so that the cut points in graph $\widetilde{H}_{min}$ are signed with the numbers $1$ and $2$; the edges are oriented in the direction of the Euler circuits' bypass. Under the equity of the previous theorem the edges of graph $\widetilde{H}_{min}$ may be signed with the pairs of three numbers ($1$, $2$ or $3$) thereby none of the adjacent edges will be colored in the same way. Let us suppose that we have two arbitrary taken edges of graph $\widetilde{H}_{min}$, which are situated on Euler circuit at the distance of one link (unit) one from another and they belong to the last faces of graph $\widetilde{H}_{min}$. And let them be colored equally.

Then there are exactly $16$ variants for the signing of graph's $\widetilde{G}_{min}$ vertexes with the numbers $1$, $2$ and $3$, so every pair of graph's $\widetilde{H}_{min}$ edges will be signed by the pair of the numbers ($1$, $2$ or $3$) and at that none of the sequentially connected edges, which are situated along Euler circuit, will be signed the same way or equally.
\end {theorem}

\begin {proof}

For the theorem to be proved it is sufficient to examine all the possible variants of signing of graph's $\widetilde{G}_{min}$ vertexes, which correspond to $9$ variants of graph's $\widetilde{H}_{min}$ bypass along Euler circuit and which will meet the theorem's conditions.

We will examine one variant, for example, \textbf{variant $2$}, more thoroughly. Graph $\widetilde{H}_{min}$ with the orientation of its edges along Euler circuit is presented in fig. \ref{Fig:8}. Graph $\widetilde{G}_{min}$ (green color) and its adjacency matrix are also presented in fig. \ref{Fig:8}. 

\begin{figure}[h]
		\includegraphics[width=0.95\textwidth]{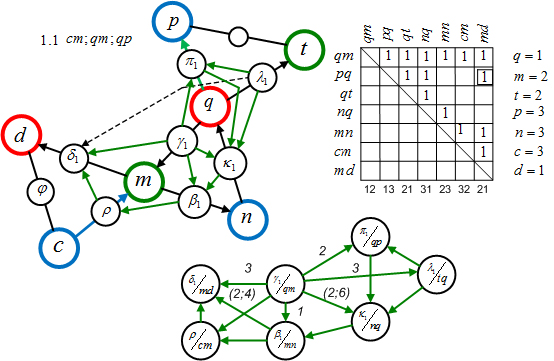}
	\caption{}
	\label{Fig:8}
\end{figure}

For a simpler visualization all the arcs in the matrix are changed into the edges. That is why the matrix has the form of the upper triangle. If the vertexes $q$ and $m$ are signed with numbers $1$ (red) and $2$ (green), then the vertex $n$ can be thereafter signed with the number $3$ (blue). Then the vertexes of graph $\widetilde{G}_{min}$ will be signed: $qm=\gamma_1=12$; $mn=\beta_1=23$; $nq=\kappa_1=31$.

The vertexes on the edges $tq$ and $md$ belong to Euler circuit (in fig. \ref{Fig:8} they are connected with the help of the dotted line) and at the same time they belong to the last faces of graph $\widetilde{H}_{min}$. That is why under the condition: $tq=md$ or $t=m=2$; $q=d=1$. Notably: $tq=21$ and $md=21$. At this condition for the signing of the $cm$ and $qp$ only the unique solution is available: $cm=32$ and $qp=13$.

So, in variant $2$ we have one unique solution. 

Let: $q=1$; $m=2$; $t=2$; $p=3$; $n=3$; $c=3$; $d=1$.

Then: $qm=12$; $qp=13$; $tq=21$; $nq=31$; $mn=23$; $cm=32$; $md=21$. 

So we'll have: $dc=13$; $cd=31$ or $pt=32$; $tp=23$.

\textbf{Variant $1$} (fig. \ref{Fig:9}): $q=1;m=2;n=3$. 

Hence: $qm=12;nq=31;mn=23$. Next: $tq=mc$, so, $t=m=2;q=c=1$.

Therefore: $tq=21$; $mc=21$. Notation of $dm$ and $qp$ is the following: $dm=32$; $qp=13$. So, for variant $1$ we have: $q=1$; $m=2$; $t=2$; $p=3$; $n=3$; $c=1$; $d=3$; Then: $qm=12$; $qp=13$; $tq=21$; $nq=31$; $mn=23$; $mc=21$; $dm=32$.
 
Finally we'll have: $dc=31$; $pt=32$ or $cd=13$; $tp=23$.

\begin{figure}[h]
		\includegraphics[width=0.95\textwidth]{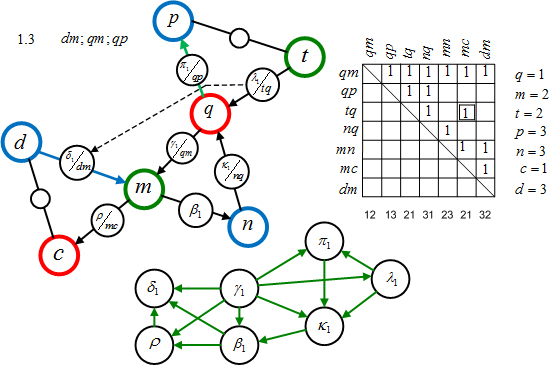}
	\caption{}
	\label{Fig:9}
\end{figure}

\textbf{Variant $3.1$} (fig. \ref{Fig:10}): $q=1$; $m=2$; $n=3$; Hence: $qm=12$; $nq=31$; $nm=32$. Next: $tq=mc=21$.

Thereafter: $qp=13$; $md=23$. 

So, for variant $3.1$ we'll have:

Let: $q=1$; $m=2$; $t=2$; $p=3$; $n=3$; $c=1$; $d=3$. Then: $qm=12$; $qp=13$; $tq=21$; $nq=31$; $nm=32$; $mc=23$; $md=21$.

Finally we'll have: $dc=31$; $cd=13$ or $pt=32$; $tp=23$.

\textbf{Variant $3.2$} (fig. \ref{Fig:10}): $q=1$; $m=2$; $n=3$. 

Hence: $qm=12$; $nq=31$; $nm=32$. Next, $tq=md$, notably, $t=m=2$; $q=d=1$; $tq=md=21$. 

Therefore, $qp=13$, and $mc=23$. 

So for variant $3.2$ we have:

Let: $q=1$; $m=2$; $t=2$; $p=3$; $n=3$; $c=3$; $d=1$. Then: $qm=12$; $qp=13$; $tq=21$; $nq=31$; $nm=32$; $mc=23$; $md=21$.
 
Finally we'll have: $dc=13$; $cd=31$ or $pt=32$; $tp=23$.

\begin{figure}[h]
		\includegraphics[width=0.95\textwidth]{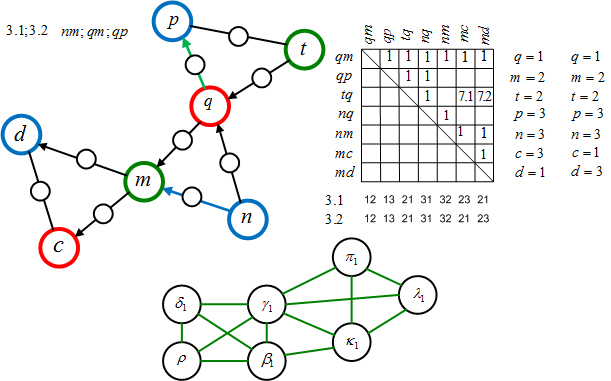}
	\caption{}
	\label{Fig:10}
\end{figure}

\textbf{Variant $4$} (fig. \ref{Fig:11}): $q=1$; $m=2$; $n=3$. 

Hence: $qm=12$; $nq=31$; $mn=23$. Next: $pq=mc$. Thereafter: $p=m=2$; $q=c=1$; $pq=mc=21$. 

The notation of $qt$ and $dm$ will be the following: $qt=13$; $dm=32$. 

For variant $4$ we have: 

Let: $q=1$; $m=2$; $t=3$; $p=2$; $n=3$; $c=1$; $d=3$. Then: $qm=12$; $pq=21$; $qt=13$; $nq=31$; $mn=23$; $mc=21$; $dm=32$. 

Finally we'll have: $dc=31$; $cd=13$ or $pt=23$; $tp=32$.

\begin{figure}[h]
		\includegraphics[width=0.95\textwidth]{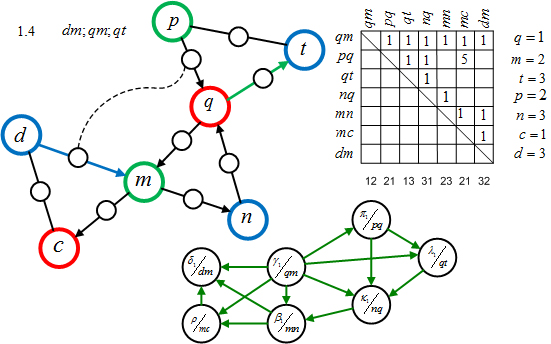}
	\caption{}
	\label{Fig:11}
\end{figure}

\textbf{Variant $5$} (fig. \ref{Fig:12}): $q=1$; $m=2$; $n=3$. Hence: $qm=12$; $mn=23$; $nq=31$; $pq=md$.

Thereafter: $p=m=2$; $q=d=1$; $pq=md=21$. The notation of $cm$ and $qt$ will be the following: $cm=32$; $qt=13$. 

So, for variant $5$ we have:

Let: $q=1$; $m=2$; $t=3$; $p=2$; $n=3$; $c=3$; $d=1$. Then: $qm=12$; $pq=21$; $qt=13$; $nq=31$; $mn=23$; $cm=32$; $md=21$.

Finally we'll have: $dc=13$; $cd=23$ or $pt=23$; $tp=32$.

\begin{figure}[h]
		\includegraphics[width=0.95\textwidth]{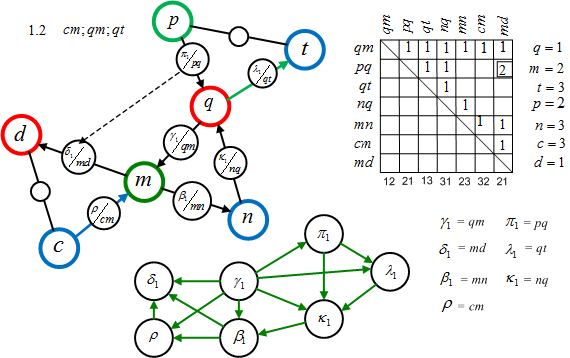}
	\caption{}
	\label{Fig:12}
\end{figure}

\textbf{Variant $6.1$} (fig. \ref{Fig:13}): $q=1$; $m=2$; $n=3$; $qm=12$; $nq=31$; $nm=32$. Next: $pq=mc=21$. Thereafter, $qt=13$; $md=23$. 

So, for variant $6.1$ we have:

Let: $q=1$; $m=2$; $t=3$; $p=2$; $n=3$; $c=1$; $d=3$. Then: $qm=12$; $pq=21$; $qt=13$; $nq=31$; $nm=32$; $mc=21$; $md=23$.

Finally we'll have: $dc=31$; $cd=13$ or $pt=23$; $tp=32$.

\begin{figure}[h]
		\includegraphics[width=0.95\textwidth]{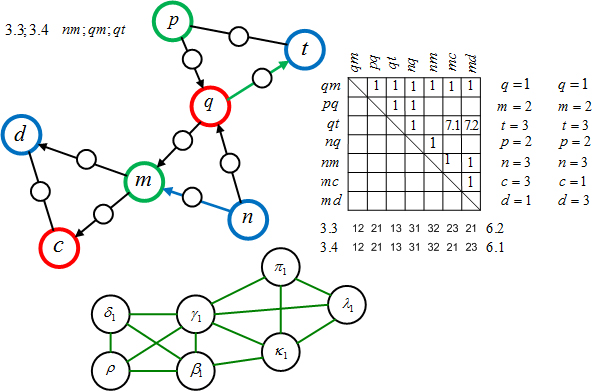}
	\caption{}
	\label{Fig:13}
\end{figure}

\textbf{Variant $6.2$} (fig. \ref{Fig:13}): $q=1$; $m=2$; $n=3$. $qm=12$; $nq=31$; $nm=32$. Next: $pq=md=21$. Then: $qt=13$ and $mc=23$. 

So, for variant $6.2$, we have:

Let: $q=1$; $m=2$; $t=3$; $p=2$; $n=3$; $c=3$; $d=1$. Then: $qm=12$; $pq=21$; $qt=13$; $nq=31$; $nm=32$; $mc=23$; $md=21$.

Finally we'll have: $dc=13$; $cd=31$ or $pt=23$; $tp=32$.

\textbf{Variant $7.1$} (fig. \ref{Fig:14}): $q=1$; $m=2$; $n=3$. 

Hence: $qm=12$; $qn=13$; $mn=23$. Next: $tq=mc=21$. Thereafter: $pq=31$ and $dm=32$. 

So, for variant $7.1$ we have:

Let: $q=1$; $m=2$; $t=2$; $p=3$; $n=3$; $c=1$; $d=3$. Then: $qm=12$; $pq=31$; $tq=21$; $qn=13$; $mn=23$; $mc=21$; $dm=32$. 

Finally we'll have: $dc=31$; $cd=13$ or $pt=32$; $tp=23$.

\begin{figure}[h]
		\includegraphics[width=0.95\textwidth]{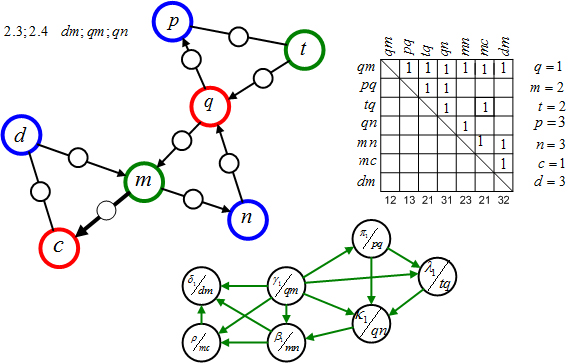}
	\caption{}
	\label{Fig:14}
\end{figure}

\textbf{Variant $7.2$} (fig. \ref{Fig:14}): $q=1$; $m=2$; $n=3$. 

Hence: $qm=12$; $qn=13$; $mn=23$. Next: $pq=mc=21$. Thereafter: $tq=31$; $dm=32$. 

So, for variant 7.2 we have the following relations:
 
Let: $q=1$; $m=2$; $t=2$; $p=3$; $n=3$; $c=1$; $d=3$. Then: $qm=12$; $pq=21$; $tq=31$; $qn=13$; $mn=23$; $mc=21$; $dm=32$. 

Finally we'll have: $dc=31$; $cd=13$ or $pt=32$; $tp=23$.

\textbf{Variant $8.1$} (fig. \ref{Fig:15}): $q=1$; $m=2$; $n=3$. 

Hence: $qm=12$; $qn=13$; $mn=23$. Next: $tq=md=21$. Thereafter: $cm=32$; $pq=31$. 

So, for variant $8.1$ we have:

Let: $q=1$; $m=2$; $t=2$; $p=3$; $n=3$; $c=3$; $d=1$. 

Then: $qm=12$; $pq=31$; $tq=21$; $qn=13$; $mn=23$; $cm=32$; $md=21$.

Finally we'll have: $cd=31$; $dc=13$ or $pt=32$; $tp=23$.

\begin{figure}[h]
		\includegraphics[width=0.95\textwidth]{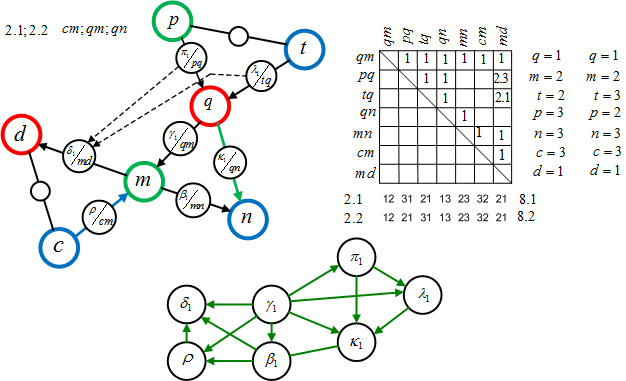}
	\caption{}
	\label{Fig:15}
\end{figure}

\textbf{Variant $8.2$} (fig. \ref{Fig:15}): $q=1$; $m=2$; $n=3$. 

Hence: $qm=12$; $qn=13$; $mn=23$. Next: $pq=md=21$. Thereafter: $tq=31$; $cm=32$. 

So, for variant $8.2$ we have:

Let: $q=1$; $m=2$; $t=3$; $p=2$; $n=3$; $c=3$; $d=1$. 

Then: $qm=12$; $pq=21$; $tq=31$; $qn=13$; $mn=23$; $cm=32$; $md=21$.

Finally we'll have: $cd=31$; $dc=13$ or $pt=23$; $tp=32$.

\textbf{Variant $9.1$} (fig. \ref{Fig:16}): $q=1$; $m=2$; $n=3$. 

Hence: $qm=12$; $qn=13$; $nm=32$. Next: $tq=mc=21$. Thereafter: $pq=31$; $md=23$. 

So, for variant $9.1$ we have:

Let: $q=1$; $m=2$; $t=3$; $p=2$; $n=3$; $c=1$; $d=3$. Then: $qm=12$;$pq=31$; $tq=21$; $qn=13$; $nm=32$; $mc=21$; $md=23$. 

Finally we'll have: $dc=31$; $cd=13$ or $pt=23$; $tp=32$.

\begin{figure}[h]
		\includegraphics[width=0.95\textwidth]{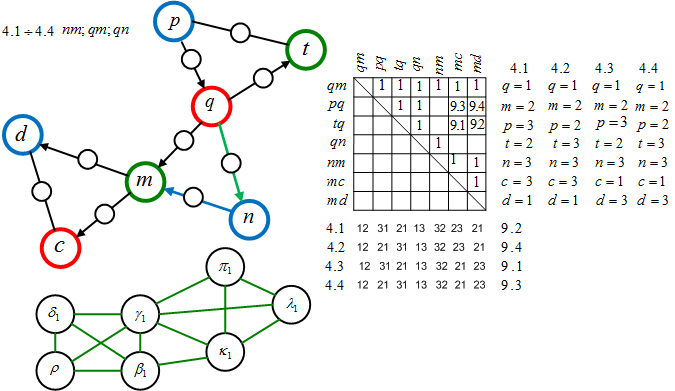}
	\caption{}
	\label{Fig:16}
\end{figure}

\textbf{Variant $9.2$} (fig. \ref{Fig:16}): $q=1$; $m=2$; $n=3$. 

Hence: $qm=12$; $qn=13$; $nm=32$. Next: $tq=md=21$. Thereafter: $pq=32$; $mc=23$. 

So, for variant $9.2$ we have:

Let: $q=1$; $m=2$; $p=3$; $n=3$; $c=3$; $d=1$. Then: $qm=12$; $pq=31$; $tq=21$; $qn=13$; $nm=32$; $mc=23$; $md=21$.

Finally we'll have: $dc=13$; $cd=31$ or $pt=32$; $tp=23$

\textbf{Variant $9.3$} (fig. \ref{Fig:16}): $q=1$; $m=2$; $n=3$. 

Hence: $pq=mc=21$. Thereafter: $tq=31$; $md=23$.

Variant $9.4$ (fig. \ref{Fig:16}): $q=1$; $m=2$; $n=3$. 

Hence: $pq=md=21$. Thereafter: $tq=31$; $mc=23$.

Let us summarize the results according to variants $9.1$, $9.2$, $9.3$ and $9.4$ in table (fig \ref{Fig:17}).

\begin{figure}[h]
		\includegraphics[width=0.55\textwidth]{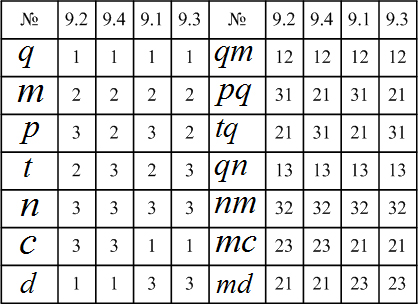}
	\caption{}
	\label{Fig:17}
\end{figure}

According to the mentioned results it is easy to be sure that all the possible variants are checked. The total number of the received variants is equal to $16$. 
\end {proof}

For more comfortable visualization let us present the received results in tables (fig. \ref{Fig:18}, \ref{Fig:19}).

\begin{figure}[h]
		\includegraphics[width=0.85\textwidth]{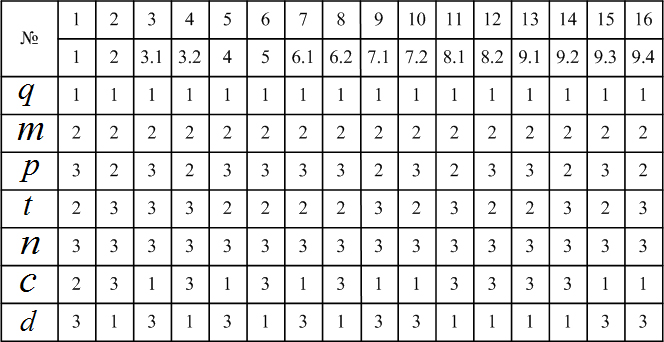}
	\caption{}
	\label{Fig:18}
\end{figure}

\begin{figure}[h]
		\includegraphics[width=0.85\textwidth]{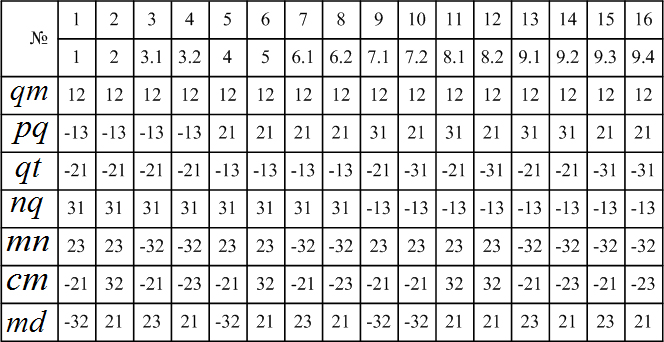}
	\caption{}
	\label{Fig:19}
\end{figure}

\textit{Remarks}:

\begin {itemize}
\item 
The sign minus (--) in table(fig. \ref{Fig:19}) is situated in the places, where the direction of the bypass changes.
\item 
The coloring is used instead of signs 1, 2 and 3 for more comfortable visualization of the graph $\widetilde{H}_{min}$ and its conjugated graph $\widetilde{G}_{min}$.
\item 
The notation of colors: 1 -- red; 2 -- green; 3 -- blue.
\end {itemize}

All the results are summerized in table (fig. \ref{Fig:20}).

\begin{figure}[h]
		\includegraphics[width=0.99\textwidth]{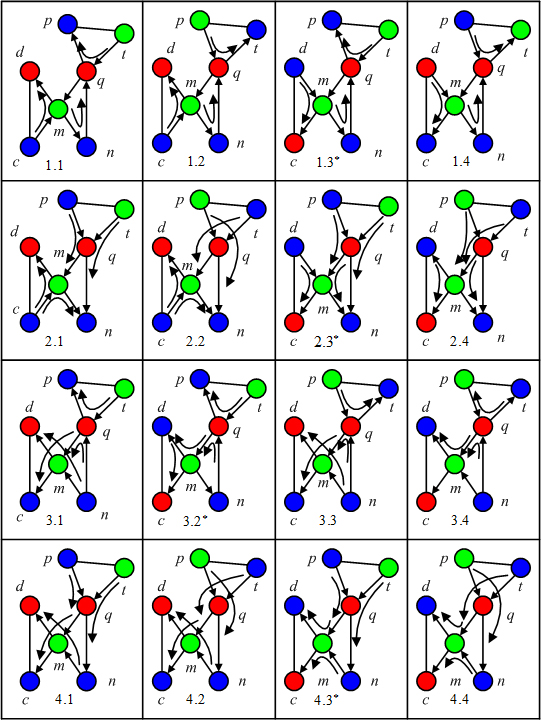}
	\caption{}
	\label{Fig:20}
\end{figure}

It's quite possible for the notation of the variants in the tables and in the text to differ a little bit. It depends on the graphical representation in colors and in settings of the signs ($1$, $2$ and $3$). But all the differences can be easily seen in the tables and the figures.

Having finished the researches on the graphs of the special type, let us examine the number properties of the planar graphs.

\section {Conclusions:}

\begin {enumerate}

\item 
A possibility of the examining of some properties of the planar triangulation $L$ on the base of special minimal graphs $H$ is shown.

\begin {itemize}

\item 
There can be five possible variants of Euler cycles in the planar conjugated triangulation (closed-loop) at the given direction of the pass through one of the external vertexes.

\item 
There can be nine possible variants of Euler cycles in the planar conjugated triangulation (open-loop) at the given direction of the pass through one of the external vertexes.

\end {itemize}

\item 
The examining of graph $G$, which is the adjacent graph to the minimal conjugated graph $H$ (the second conversion of graph $L$), indicated that its vertexes can be colored correctly in three colors.

\item 
But in order to confirm the correctness of the coloring $16$ variants of the pass through the minimal graph must be checked.

\item 
Such researches on the minimal graphs can help in the improvement of the complexity of the algorithms, which may be constructed in order to achieve the correct coloring.

\end {enumerate}


\begin{thebibliography}{9}                                                                              
\bibitem {Malinin} \textsc{Malinin, L., Malinina, N.}, \textit{Graph isomorphism in theorems and algorithms}, LIBROCOM, Moscow. 2009

\bibitem {Malinina} \textsc {Malinin, L., Malinina, N.}, \textit {On the soluton of the Graph Isomorphism Problem, Part I}, ArXiv (Cornell University Library), http://arxiv.org/abs/1007.1059, 2010

\bibitem {Malinina1} \textsc {Malinin, L., Malinina, N.}, \textit {A converting of the directed graphs}, ArXiv (Cornell University Library), http://arxiv.org/abs/1210.6088, 29p., (2012)


\end{thebibliography}
\end{document}